\documentclass[a4paper,10pt,twoside,reqno,nonamelimits]{article}
\usepackage{a4wide}
\usepackage[colorlinks=true,urlcolor=blue,linkcolor=blue,citecolor=blue]{hyperref}
\usepackage{newcent}         

\usepackage{amsmath}
\usepackage{amssymb}
\usepackage{amsfonts}
\usepackage{amscd}
\usepackage{amsthm}
\usepackage{amsxtra}
\usepackage{mathrsfs}
\usepackage{nicefrac}
\usepackage{graphicx}
\usepackage{xcolor}
\usepackage{tikz}
\usepackage{enumitem}
\setlist[enumerate]{leftmargin=2em,itemindent=0em, labelindent=0pt,labelwidth=1.5em,labelsep=.5em, align=left, noitemsep}
\newlist{txtenum}{enumerate}{1}
\setlist[txtenum]{leftmargin=0em,itemindent=1.5em, labelindent=0pt,labelwidth=1em,labelsep=.5em, align=left}

\theoremstyle{plain}
\newtheorem{theorem}{Theorem}
\newtheorem*{theorem*}{Theorem}

\newtheorem*{proposition*}{Proposition}

\newtheorem*{corollary*}{Corollary}

\newtheorem{lemma}[theorem]{Lemma}
\newtheorem*{lemma*}{Lemma}

\newtheorem*{observation*}{Observation}

\newtheorem*{conjecture*}{Conjecture}

\newtheorem*{question*}{Question}

\newtheorem*{questions*}{Questions}

\newtheorem*{problem*}{Problem}

\newtheorem*{problems*}{Problems}

\newtheorem*{openproblem*}{Open Problem}

\theoremstyle{definition}

\newtheorem*{definition*}{Definition}

\newtheorem*{example*}{Example}

\newtheorem*{exercise*}{Exercise}

\newtheorem{remark}[theorem]{Remark}
\newtheorem*{remark*}{Remark}

\newtheorem*{remarks*}{Remarks}

\theoremstyle{remark}

\newtheorem*{claim*}{Claim}

\newcommand{\subclass}[1]{}

\newcommand{\enumTi}[1]{\renewcommand{\theenumi}{#1}}

\newcommand{\alphenumi}{\enumTi{\alph{enumi}}}

\newcommand{\romenumi}{\enumTi{\roman{enumi}}}

\newlength{\hspaceforlengthglumpf}

\renewcommand{\em}{\sl}

\newcommand{\lt}{\left}
\newcommand{\rt}{\right}

\newcommand{\abs}[1]{{\lt\lvert{#1}\rt\rvert}}

\newcommand{\NN}{\mathbb{N}}

\newcommand{\RR}{\mathbb{R}}

\newlength{\algotabbingwidth}
\newcommand{\Sets}{\mathcal S}
\newcommand{\Trees}{\mathcal T}
\newcommand{\Proofs}{\mathcal C}

\newcommand{\Code}[1]{\texttt{\textsc{#1}}}
\definecolor{mygray}{gray}{0.33333}
\newcommand{\No}[1]{\textcolor{mygray}{\tt #1}}
\begin{document}
\title{On the Combinatorial Lower Bound for the Extension Complexity of the Spanning Tree Polytope}
\author{Kaveh Khoshkhah, Dirk Oliver Theis\thanks{Supported by the Estonian Research Council, ETAG (\textit{Eesti Teadusagentuur}), through PUT Exploratory Grant \#620, and by the European Regional Development Fund through the Estonian Center of Excellence in Computer Science, EXCS.}\\[1ex]
  \small Institute of Computer Science {\tiny of the } University of Tartu\\
  \small \"Ulikooli 17, 51014 Tartu, Estonia\\
  \small \texttt{\{kaveh.khoshkhah,dotheis\}@ut.ee}%
}

\date{Sat Jan 21 15:17:26 EET 2017}
\maketitle

\begin{abstract}
  In the study of extensions of polytopes of combinatorial optimization problems, a notorious open question is that for the size of the smallest extended formulation of the Minimum Spanning Tree problem on a complete graph with~$n$ nodes.  The best known lower bound is the trival (dimension) bound, $\Omega(n^2)$, the best known upper bound is the extended formulation by Wong (1980) of size $O(n^3)$ (also Martin, 1991).

  In this note we give a nondeterministic communication protocol with cost $\log_2(n^2\log n)+O(1)$ for the support of the spanning tree slack matrix.
  This means that the combinatorial lower bounds can improve the trivial lower bound only by a factor of (at most) $O(\log n)$.

  \par\medskip%
\end{abstract}

\section{Introduction}\label{sec:intro}
The Spanning Tree polytope, $P_n$, has as its vertices the characteristic vectors in $\RR^{\binom{[n]}{2}}$ of edge-sets of trees with node set $[n]:=\{1,\dots,n\}$ (we use binomial coefficient notation for sets of subsets).  A complete system of inequalities and equations was given by Edmonds~\cite{Edmonds:submodular:1970}:
\begin{subequations}
  \begin{align}
    \sum_{e\in\binom{[n]}{2}} x_e &=   n-1           &&\label{eq:Edmonds:eqn}\\
    \sum_{e\in\binom{S}{2}} x_e   &\le \abs{S}-1     &&\forall S\subset [n],\ \abs{S}>1 \label{eq:Edmonds:cyc}\\
    x_e                           &\ge 0             &&\forall e\in\binom{[n]}{2}. \label{eq:Edmonds:nng}
  \end{align}
\end{subequations}
This system has exponentially many facet-defining inequalities.  There is a classical extended formulation by Wong/Martin~\cite{Wong:tsp:1980,Martin:91} with $O(n^3)$ inequalities (and variables).  A notorious open problem in polyhedral combinatorial optimization, highlighted by M.~Goemans at the 2010 \textit{Carg\`ese Workshop on Combinatorial Optimization,} asks whether or not an extended formulation with $o(n^3)$ inequalities exists.

There has been some progress on sub-trees in specific graph classes instead of the complete graph (see, e.g., \cite{Fiorini-Huynh-Joret-Pashkovich:spanningtree:2017,pashkovich:phd:2012} and the references therein), but there does not seem to be a compelling reason to believe that there exists an extended formulation with $o(n^3)$ inequalities in the setting of the complete graph, as described here.

The only known lower bound is $\Omega(n^2)$ --- a ``trivial'' lower bound (it is the dimension of the spanning tree polytope, $P_n$).

The smallest number of inequalities in an extended formulation is called the extension complexity.  More formally and generally, let~$P \subset \RR^d$ be a polytope.  A polytope~$Q \subset \RR^{e}$ is called an \textit{extension of~$P$}, if there exists a projective mapping $\pi\colon\RR^e\to\RR^d$ which maps $P$ \textsl{onto}~$Q$.  This allows to reduce linear programming over~$P$ to linear programming over~$Q$.  The \textit{size of the extension} is the number of facets of~$Q$, and the \textit{extension complexity}~\cite{Fiorini-Kaibel-Pashkovich-Theis:CombLB:13} of~$P$ is the smallest size of an extension of~$P$.

There are links between extension complexity and communication complexity, a fact which has been observed and used by Yannakakis~\cite{Yannakakis:91}, and recently strengthened by Faenza et al.~\cite{Faenza-Fiorini-Grappe-Tiwary:nngrk-rCC:2012}.  One of these links is the following.  Denoting by $F(P)$ and $V(P)$ the set of facets and vertices, respectively, of the polytope~$P$, let $f_P \colon F(P)\times V(P) \to \{0,1\}$ be the boolean function which maps a pair of a facet and a vertex to~$0$, if the vertex lies on the facet, and to~$1$ otherwise.  Then the nondeterministic communication complexity of $f_P$ is a lower bound for the binary logarithm of the extension complexity of~$P$~\cite{Yannakakis:91}.  Nondeterministic communication complexity can be defined as the binary logarithm of the so-called \textit{rectangle covering number,} a combinatorial concept, but in this paper, we stick to the terminology of communication complexity.\footnote{We do that for employment purposes: Mainly of Alice \& Bob (who would otherwise be out of work), but also of the authors (because communication complexity is currently so much easier to sell than combinatorics).}

Lower bounds based on nondeterministic communication complexity have been successful for several families of polytopes of combinatorial optimization problems, e.g., the Bipartite Matching polytopes, Traveling Salesman polytopes, Cut polytopes, Stable Set polytopes (see \cite{Fiorini-Massar-Pokutta-Tiwary-Dewolf:ACM:15} for more examples).

For Spanning Tree polytopes, we can disregard the $O(n^2)$ nonnegativity inequalities~\eqref{eq:Edmonds:nng} (see next section).
Defining
\begin{align*}
  \Sets  &:= \Bigl\{ S \subsetneq [n]             \Bigm| \abs{S}>1       \Bigr\} \text{ and} \\
  \Trees &:= \Bigl\{ T \subseteq  \binom{[n]}{2}  \Bigm| ([n],T) \text{ tree } \Bigr\},
\end{align*}
(we use the notation $(V,E)$ for a graph with node set~$V$ and edge set~$E$), the resulting boolean function can be written as
\begin{equation}\label{eq:def-sptree-f}
  f_n\colon \Sets\times \Trees \colon
  (S,T) \mapsto
  \begin{cases}
    1, &
    \begin{aligned}[t]
      &\text{if the sub-forest of~$T$ induced by~$S$,}\\
      &\text{$(S,T\cap \textstyle\binom{S}{2})$, is disconnected;}\\[.5ex]
    \end{aligned}%
    \\[.5ex]
    0, &
      \begin{aligned}[t]
        &\text{if the sub-forest of~$T$ induced by~$S$ is a tree,}\\
        &\text{i.e.,  $(S,T\cap \textstyle\binom{S}{2})$ is connected.}
      \end{aligned}
  \end{cases}
\end{equation}
Indeed, a tree~$T$ is on a facet defined by an inequality of the type~\eqref{eq:Edmonds:cyc}, if and only if that inequality is satisfied with equation when plugging in the characteristic vector~$x^T$ of~$T$ (meaning $x^T_e = 1$ iff $e\in T$, otherwise~$0$), which is the case if and only if the sub-forest of~$T$ induced by~$S$ is a tree.

An $3\log_2 n + O(1)$ upper bound for the nondeterministic communication complexity of the Spanning Tree polytope follows from the existence $O(n^3)$ extended formulation, and a nondeterministic communication protocol with that cost can be readily written down (see next section).

Over the last 6 years, many a fingernail was gnawed when researchers (including the authors) attempted to prove a non-trivial lower bound for the extension complexity of the Spanning Tree polytope through nondeterministic communication complexity.  The binary logarithm of the dimension of any polytope~$P$ is a trivial lower bound to the nondeterministic communication complexity of~$f_P$ (not just to the binary logarithm of the extension complexity).  For the Spanning Tree polytope, this amounts to $2\log_2 n - O(1)$, and nothing better is known.
Weltge~\cite{Weltge:phd:2015} made progress by proving an upper bound of $\frac{8}{3}\log_2 n + \log_2\log_2 n +O(1)$ for a very important \textsl{lower bound} to the nondeterministic communication complexity: the (binary logarithm of the) \textsl{fractional} rectangle covering number.  However, no new upper bound to the nondeterministic communication complexity can be derived from Weltge's result.  Recently, another convenient lower bound to the nondeterministic communication complexity, the so-called fooling-set bound was proved to be useless~\cite{Khoshkhah-Theis:treeFool:2017} (a result that was not so surprising, seeing as the fooling-set bound of a ``typical'' boolean function appears to grow at most slightly faster than the dimension~\cite{Pourmoradnasseri-Theis:foolrk:arXiv}).

In this note, we give an efficient nondeterministic communication protocol for $f_n$, which implies the following upper bound for $f_{P_n}$, the boolean function associated with the Spanning Tree polytope.

\begin{theorem}\label{thm:main}
  The nondeterministic communication complexity of $f_{P_n}$ is $2\log_2 n + \log_2\log_2(n) +O(1)$.
\end{theorem}

\paragraph{The remainder of this paper is organized as follows.}  %
In Section~\ref{sec:basics}, we review some basic definitions and lay the ground on which our nondeterministic communication protocol, described in Section~\ref{sec:prot}, is based.  In Section~\ref{sec:proof}, we prove the correctness of the protocol.  The paper closes with a short discussion, in Section~\ref{sec:conclusion}, of the new status quo on the extension complexity of the Spanning Tree polytope.

\section{The Trivial Bound on the Nondeterministic Communication Complexity of the Spanning Tree Polytope}\label{sec:basics}
In an attempt to make this note accessible to the non-expert in communication complexity, we briefly review the definition of nondeterministic communication complexity.

Let $f\colon X\times Y \to \{0,1\}$ be a boolean function. In Communication Complexity, Alice and Bob are tasked with computing the value $f(x,y)$, when Alice and Bob each know only part of the input: Alice gets~$x$, and Bob gets~$y$.  They have to communication in order to determine~$f(x,y)$.  Full knowledge of~$f$ and unlimited computational power are assumed.  In nondeterministic communication complexity, there is, in addition, a \textit{Prover,} who tries to convince Alice and Bob that the output is~$1$: the Prover will send a \textit{certificate} to Alice and Bob, based on which they must make a decision.

A \textit{nondeterministic communication protocol} consists of a set $\Proofs$ of possible certificates that the Prover can send, together with description of how Alice and Bob react, based on their respective inputs, to the certificate sent by the Prover.  Alice and Bob can communicate (i.e., send/receive bits) with each other (although in the protocols in this paper, they don't, so we hand-wave that part of the definition of a nondeterministic communication protocol).  Ultimately, Alice and Bob each either \textit{accepts} or \textit{rejects} their respective inputs based on the certificate sent by the Prover.

Such a protocol \textit{computes~$f$,} if:
\begin{enumerate}[label=(\roman*),nosep]
\item\label{def:ndCP-correct:1} For each input $(x,y)$ with $f(x,y)=1$, there is a certificate $C\in\Proofs$ such that, if the Prover sends~$C$, then Alice and Bob \textsl{both} accept;
\item\label{def:ndCP-correct:0} For each input $(x,y)$ with $f(x,y)=0$, for every certificate $C\in\Proofs$, if the Prover sends~$C$, then at least one of Alice and Bob rejects.
\end{enumerate}
Informally, the way we talk about the Prover is that his goal is to make Alice and Bob accept.  Knowing Alice's and Bob's parts of the protocol, if the input $(x,y)$ is such that $f(x,y)=1$, he is honest, i.e., he sends a certificate which really proves that $f(x,y)=1$ in a way agreed to between the three parties.  If, however, the input $(x,y)$ is such that $f(x,y)=0$, the Prover has no chance but to lie, and he does so in a way that will fool Alice and Bob into accepting, if that is possible.

The \textit{cost} of a protocol is number of bits sent by Alice and Bob, plus $\log_2(\abs{C})$, the (idealized, since possibly fractional) number of bits sent by the Prover.

The \textit{nondeterministic communication complexity} of a function~$f$ is the smallest cost of a protocol computing~$f$.

\begin{remark}
  Wlog, Alice and Bob do not communicate among themselves: the Prover could simply send the messages they would be exchanging, which they would verify.  In that case, it is easy to see that, for every $C\in\Proofs$, the set of $(x,y)\in X\times Y$ for which Alice and Bob both accept is of the form $K\times L$ --- a \textit{rectangle.}  The protocol is correct, if (1) no such rectangle contains an input $(x,y)$ with $f(x,y)=0$, and (2) every input $(x,y) \in X\times Y$ with $f(x,y)=1$ is contained in one such rectangle.  Hence, the nondeterministic communication complexity is equal to the $\log_2$ of the minimum number of $1$-rectangles needed to cover all $1$-inputs.
\end{remark}

\subsubsection*{The $O(n^3)$ Protocol for Spanning Tree}
As an example, we consider $X := \Sets$, $Y := \Trees$, and $f_n$ as defined in~\eqref{eq:def-sptree-f}.  So Alice will get a set $S\in\Sets$, and Bob will get a tree $T\in\Trees$, and they should both accept if $T$ is disconnected on~$S$.  We set
\begin{equation*}
  \Proofs := [n]^3 = \bigl\{ (u,t,v) \mid u,t,v \in [n] \bigr\}.
\end{equation*}
Figure~\ref{fig:n3-prot} describes Alice's and Bob's parts of the protocol.
\begin{figure}[thbp]
  \begin{flushleft}%
    \begin{tabular}[t]{|llc|}
      \hline
      \multicolumn{2}{|l}{\textbf{Alice:}}                                      &\\
      \cline{1-2}
      \No{1.}&Let $S \in \Sets$ be Alice's input.                        &\\
      \No{2.}&Let $(u,t,v) \in\Proofs$ be the triple sent by the Prover. &\\
      \No{3.}&\Code{If} $u,v \in S$, $t\not\in S$, \Code{Accept};                      &\\
      \No{4.}&\Code{Else} \Code{Reject}.                                               &\\
      \hline
    \end{tabular}
  \end{flushleft}
  \begin{flushright}
    \begin{tabular}[t]{|llc|}
      \hline
      \multicolumn{2}{|l}{\textbf{Bob:}}                                           &\\
      \cline{1-2}
      \No{1.}&Let $T \in \Trees$ be Bob's input.                                   &\\
      \No{2.}&Let $(u,t,v) \in\Proofs$ be the triple sent by the Prover.           &\\
      \No{3.}&\Code{If} $t$ is on the path in~$T$ between $u$ and~$v$, \Code{Accept};            &\\
      \No{4.}&\Code{Else} \Code{Reject}.                                                         &\\
      \hline
    \end{tabular}
  \end{flushright}
  \caption{$O(n^3)$ protocol for Spanning Tree (Alice-Bob part)}\label{fig:n3-prot}
\end{figure}

It is fairly obvious that the protocol computes~$f_n$, but we take the opportunity to make a definition that we will need later.  Given $(S,T)\in\Sets\times\Trees$, we say that a triple $(u,t,v)\in[n]^3$ is a \textit{witness for $f_n(S,T)=1$,} if the conditions in the protocol in Figure~\ref{fig:n3-prot} hold, i.e., if:
\begin{enumerate}[label=(\Alph*),nosep]
\item\label{def:witness:Alice} $u,v \in S$, $t\not\in S$; and
\item\label{def:witness:Bob} $t$ is on the path in~$T$ between $u$ and~$v$.
\end{enumerate}
The terminology makes sense: For ever $(S,T)\in\Sets\times\Trees$, we have $f_n(S,T)=1$, if and only if a witness for $f_n(S,T)=1$ exists.  Indeed, the sub-forest of~$T$ induced by~$S$ is disconnected, if and only if there is a pair of nodes $u,v$ such that the (unique) path between $u$ and~$v$ in~$T$ leads through a node~$t$ which is not in~$T$.

Hence, on the one hand, the Prover can accurately prove that to Alice and Bob that $f_n(S,T)=1$ by sending a witness for that fact as certificate.  On the other hand, if $f_n(S,T)=0$, no triple forms a witness, so by verifying the two conditions, Alice and Bob can refute the certificate sent by the Prover.  The key property of the conditions in the context of communication complexity is that Alice and Bob can verify their respective parts of the condition independently by only looking at their own input.

Note that the definition of witness is symmetric in $u,v$: $(u,t,v)$ is a witness for $f_n(S,T)=1$ iff $(v,t,u)$ is one.  Clearly, if $u=v$, $(u,t,v)$ is never a witness for anything.

\subsubsection*{Nonnegativity Inequalities}
For the sake of completeness, we sketch the argument why the $\binom{n}{2}$ nonnegativity inequalities~\eqref{eq:Edmonds:nng} can be omitted for the upper bound on $f_{P_n}$.

\begin{lemma}[Folklore]
  Let $X = X^0 \cup X^1$ with $X^0\cap X^1 =\emptyset$, let $f^i \colon X^i \times Y\to \{0,1\}$ be boolean functions, and let $f\colon X\times Y$ be defined through $f(x,y) = f^0(x,y)$, if $x\in X^0$ and $f(x,y) = f^1(x,y)$, if $x\in X^1$.  Then nondeterministic communication protocols for $f^0$ and~$f^1$ can be combined to form a protocol for~$f$ whose cost is at most $1$ plus the maximum of the costs of the protocols for $f^0$ and~$f^1$.
\end{lemma}
\begin{proof}[Sketch of Proof]
  To certify that $f(x,y)=1$, the Prover first sends one bit $i\in\{0,1\}$, signifying that $x\in X^i$, then he sends the certificate for that case.  Alice can check whether the Prover lies in the first bit, and rejects if he does, otherwise proceeds as in the corresponding protocol.  Bob follows the protocol indicated by~$i$.
\end{proof}

Combining the lemma with the fact that the nondeterministic communication complexity of $f\colon X\times Y\to\{0,1\}$ is at most $\log_2(\abs{X})$ (the Prover can send $x$), when adding the nonnegativity inequalities, we obtain a protocol which uses at most~$1$ more bit than the one described in the next section.  This bit is swalloed in the $O(1)$-term of Theorem~\ref{thm:main}.

\section{A Parsimonious Protocol}\label{sec:prot}
(From now on, we abbreviate $f_n$ to~$f$.)  
The protocol in the previous section requires the Prover to send one of $O(n^3)$ certificates.  To reduce that number, the fundamental intuition is to perform a ``lossy compression'' of the witness: some information is lost, but Alice and Bob can still make their decisions.  This only works if the certificate which the Prover sends on input $(S,T)$ in the case $f(S,T)=1$ are carefully chosen.

To describe the Prover's message, we need the following definition.  Consider $u,v\in [n]$ with $u<v$.  We say that $v$'s range is the set $R_v$ of numbers in $[n]$ which are closer to~$v$ than to~$u$, and $v$'s range is the set $R_v$ of numbers which are at least as close to~$u$ as they are to~$v$; in symbols:
\begin{align*}
  R_u &:= \{ j\in\NN \mid j \le (u+v)/2  \},
  \\
  R_v &:= \{ j\in\NN \mid (u+v)/2 < j \}.
\end{align*}

Now we are ready to describe the Prover's message.  Suppose Alice's input is the set~$S$ and Bob's input is the tree~$T$.  If $f(S,T)=1$, among all witnesses $(u,t,v)$ with $u<v$, the Prover chooses one which minimizes the expression
\begin{equation}\label{eq:minimize}
  \mu(u,t,v) := \abs{t-u}+\abs{t-v}.
\end{equation}
We call such a witness (satisfying $u<v$ and minimizing~\eqref{eq:minimize}) a \textit{valid} witness.  If $f(S,T)=1$, the Prover takes any valid witness and sends a quintuple $h(u,t,v)$ consisting of
\begin{itemize}[nosep]
\item the numbers $u$, and $v$;
\item one bit, $\pi$, indicating whether $t\in R_v$ (i.e., $1$, if that is the case and $0$ if it isn't);
\item one bit, $\delta$, indicating whether $t<u$, if $t\in R_u$, or $t<v$, if $t\in R_v$, respectively;
\item the number $d := \lfloor \log_2\abs{t-u} \rfloor$, if $u\in R_u$, or $d:=\lfloor \log_2\abs{t-v} \rfloor$, if $v \in R_v$, respectively.
\end{itemize}
Here is the set of certificates:
\begin{equation*}
  \Proofs
  := \Bigl\{ (u,v,\pi,\delta,d) \in [n]\times [n]\times \{ 0,1 \} \times \{0,1\} \times \{0,\dots,\lfloor \log_2 n \rfloor\}
  \Bigm|
  \text{$u<v$}
  \Bigr\}
\end{equation*}
Alice's and Bob's parts of the protocol are displayed in Figure~\ref{fig:n2lgn-prot}.  Since Alice and Bob do not communicate, the total cost of the protcol is
\begin{equation*}
  \log_2\abs{C} = 2\log_2 n + 2 + \log_2\log_2n + O(1).
\end{equation*}

\begin{figure}[thbp]
  \begin{flushleft}%
    \begin{tabular}[t]{|llc|}
      \hline
      \multicolumn{2}{|l}{\textbf{Alice:}}                                                &\\
      \cline{1-2}
      \No{1.}  &Let $S \in \Sets$ be Alice's input.                                           &\\
      \No{2.}  &Let $c:=(u,v,\pi,\delta,d) \in\Proofs$ be the certificate sent by the Prover. &\\
      \No{3.}  &\Code{If} $u \notin S$ or $v\notin S$: \Code{Reject}.                         &\\
      \No{4.}  &\Code{For\ All} $r\in[n]\setminus\{u,v\}$ with $h(u,r,v)=c$:                                  &\\
      {}       &\quad\quad\Code{If} $r\in S$: \Code{Reject}.                                  &\\
      \No{5.}  &\Code{Accept}.                                                                &\\
      \hline
    \end{tabular}
  \end{flushleft}%
  \begin{flushright}
    \begin{tabular}[t]{|llc|}
      \hline
      \multicolumn{2}{|l}{\textbf{Bob:}}                                                      &\\
      \cline{1-2}
      \No{1.}  &Let $T \in \Trees$ be Bob's input.                                                &\\
      \No{2.}  &Let $c:=(u,v,\pi,\delta,d) \in\Proofs$ be the certificate sent by the Prover.     &\\
      \No{3.}  &\Code{For\ All} $r\in[n]\setminus\{u,v\}$ with $h(u,r,v)=c$:                                      &\\
      {}       &\quad\quad\Code{If} $r$ is on the path in~$T$ between $u$ and~$v$: \Code{Accept}. &\\
      \No{4.}  &\Code{Reject}.                                                                    &\\
      \hline
    \end{tabular}
  \end{flushright}
  \caption{The parsimonious protocol for Spanning Tree (Alice-Bob part)}\label{fig:n2lgn-prot}
\end{figure}

If it \textsl{were not} for the rounding down, in~$d$, of the $\log_2$ of the distance of~$t$ to either $u$ or~$v$ (whichever is closer), the certificate data \textsl{would} allow to reconstruct~$t$ exactly: with $\tilde d := \log_2\abs{t-u}$, if $u\in R_u$, or $\tilde d := \log_2\abs{t-v}$, if $v \in R_v$, respectively, (i.e., $d = \lfloor \tilde d \rfloor$) we have
\begin{equation*}
  t =
  \begin{cases}
    u + (-1)^\delta\, 2^{\tilde d} &\text{ if $\pi=0$;} \\
    v + (-1)^\delta\, 2^{\tilde d} &\text{ otherwise.}
  \end{cases}
\end{equation*}
Sending the rounded-down~$d$ reduces the factor in front of the $\log_2 n$ in the cost of the protocol from $3$ to~$2$, but it clearly incurs a loss of information.  However, Alice and Bob can make decisions based on~$d$, in the way described in Figure~\ref{fig:n2lgn-prot}.  In the next section, we prove that their decisions are correct (in the sense that the nondeterministic communication protocol really computes~$f$).

\section{Proof of Correctness}\label{sec:proof}
We now prove the correctness of the protocol described in the previous seciton.

We first consider the condition~\ref{def:ndCP-correct:0} of the definition of a nondeterministic communication protocol computing a function.
\begin{lemma}
  Let $S\in\Sets$ be Alice's input set and $T\in\Trees$ be Bob's input tree.
  If Alice and Bob accept, then $f(S,T)=1$.
\end{lemma}
\begin{proof}
  Suppose that Bob accepts.  That means that in the loop~\No{3} of his part of the protocol, he has found an $r_0\in[n]\setminus\{u,v\}$ such that $(u,r_0,v)$ satisfies~\ref{def:witness:Bob}.

  In Alice's loop~\No{4}, she has checked the~$r_0$ found by Bob.  If Alice has accepted, that means that this $r_0 \notin S$, i.e., $(u,r_0,v)$ also satisfies~ \ref{def:witness:Bob}.  In short, $(u,r_0,v)$ is a witness for $f(S,T)=1$.
\end{proof}

We now come to the condition~\ref{def:ndCP-correct:1} of the definition of a protocol computing a function.
Fundamentally, the property of~$f$ which makes the protocol work is the ``ubiquity'' of witnesses: If $f(S,T)=1$, ``many'' witnesses exist for that fact.  The precise property we need is the following.

\begin{lemma}[Triangle Lemma]
  Let $(S,T)\in\Sets\times \Trees$, and let $u,v,w\in S$, $x\notin S$.  If $(u,t,v)$ is a witness for $f(S,T)=1$, then so is at least one of $(v,t,w)$, $(w,t,u)$.
\end{lemma}
\begin{proof}
  Since the lemma is trivially true if $\abs{\{u,v,w\}} \le 2$, we assume that $u,v,w$ are all distinct.

  Property~\ref{def:witness:Alice} is clearly satisfied by all three triples.  As for property~\ref{def:witness:Bob}, suppose $(u,t,v)$ is a witness for $f(S,T)=1$, and consider the rooted tree which results from~$T$ by choosing~$t$ as the root.  If $(u,t,v)$ is a witness, $u$ and~$v$ are descendants of two different children $s_u$, $s_v$ of~$t$, so at least one of these two children, $s'$, is not an ancestor of~$w$.  If $s'=s_u$, then the path between~$w$ and~$u$ goes through~$t$, so $(w,t,u)$ is a witness for $f(S,T)=1$; if $s'=s_v$, then $(v,t,w)$ is a witness.
\end{proof}

\begin{lemma}
  Let $S\in\Sets$ be Alice's input set and $T\in\Trees$ be Bob's input tree, and let $(u,t,v)$ be a valid witness for $f(S,T)=1$.
  If the prover sends $h(u,t,v)$, then Alice and Bob both accept.
\end{lemma}
\begin{proof}
  Let us start with Bob:  Searching through all~$r$ in the loop~\No{3}, he will encounter~$t$ and accept.

  As for Alice, we have to prove that of all the $r\in[n]\setminus\{u,v\}$ with $h(u,r,v)=h(u,t,v)$, none is in~$S$.  Here, we use the triangle lemma, and the minimality of the expression~\eqref{eq:minimize}.

  For a proof by contradiction, assume that $r\in S$ and $h(u,r,v)=h(u,t,v)$.  Since the $\pi$-entries of $h(u,r,v)$ and $h(u,t,v)$ are the same, $r$ and~$t$ are both either in $R_u$ (i.e., at least as close to $u$ as to~$v$) or in $R_v$ (i.e., closer to~$v$ than to~$u$).  Let us consider the case $\pi=0$, which indicates $r,t\in R_u$ --- the other case is similar.

  Since the $\delta$-entries of $h(u,r,v)$ and $h(u,t,v)$ are the same, $r$ and~$t$ are both either to the left of~$u$ or to the right of~$u$.  Let us assume $\delta=0$, which indicates $u < r,t$ --- the other case is similar.

  By the triangle lemma, one of $(r,t,v)$ or $(u,t,r)$ is a witness for $f(S,T)=1$.  Let us consider the case that $(r,t,v)$ is such a witness, and consider $\mu(r,t,v)$.  Since the right-most entry of the triple has not changed, the second summand in the expression~\eqref{eq:minimize} for $\mu(r,t,v)$ is the same as in $\mu(u,t,v)$.  As for the left summand, let~$d$ denote the common rightmost entry of $h(u,r,v)$ and $h(u,t,v)$.   We have
  \begin{equation}\label{eq:range-r,t}
    u+2^d \le r,t < u+2\cdot 2^d,
  \end{equation}
  and hence
  \begin{equation}\label{eq:strict-ieq-r,t,u}
    \abs{r-t} < 2^d \le \abs{t-u}.
  \end{equation}
  This means that $\mu(r,t,v) < \mu(u,t,v)$, and contradicts the condition that $(u,t,v)$ is a \textsl{valid} witness.

  In the case that $(u,t,r)$ is a witness for $f(S,T)=1$ instead of $(r,t,v)$, we also have~\eqref{eq:range-r,t}, also implying~\eqref{eq:strict-ieq-r,t,u}.  This time, the left summand in in the expression~\eqref{eq:minimize} for $\mu(u,t,r)$ is the same as the one in $\mu(u,t,v)$, but for the right summand for $\mu(u,t,r)$ is
  \begin{equation}\label{eq:oweinfosdif}
    \abs{r-t} < \abs{t-u} \le \abs{t-v},
  \end{equation}
  where the second inequality follows from the assumption (case) that the $\pi$-entry in $h(u,t,v)$ is~$0$, i.e., $t$ is closer to~$u$ than to~$v$.  The right-hand-side term in~\eqref{eq:oweinfosdif} is the second summand for $\mu(u,t,v)$.  Hence, in this case we also arrive at a statement contradicting the validity of the witness $(u,t,v)$.
\end{proof}

This concludes the proof of Theorem~\ref{thm:main}.

\section{Conclusions}\label{sec:conclusion}
Despite unrelenting interest in the problem over the last 6 years (e.g., \cite{Beasley-Klauck-Lee-Theis:Dagstuhl:13,Klauck-Lee-Theis-Thomas:Dagstuhl:15}, Carg\`ese workshop \textit{Extended Formulations II} (2014)), the extension complexity of the Spanning Tree polytope seems to be as open as ever.

To the authors, it appears as if even the slightest improvement of either the upper bound (e.g., $n^3/\log\log n$) or the lower bound (e.g., $n^2\log\log n$) to the extension complexity of the Spanning Tree polytope could be a breakthrough.  While Theorem~\ref{thm:main} determines the nondeterministic communication complexity lower bound of the Spanning Tree polytope up to a multiplicative $O(\log n)$ term, it is still conceivable that that method could yield a lower bound of $\Omega(n^2\log n)$.  However, it appears more promising to focus on the non-combinatorial bounds (e.g., \cite{Rothvoss14}).

%
%
%

\subsubsection*{Acknowledgments}
This research was supported by the Estonian Research Council, ETAG (\textit{Eesti Teadusagentuur}), through PUT Exploratory Grant \#620.  We also gratefully acknowledge funding by the European Regional Development Fund through the Estonian Center of Excellence in Computer Science, EXCS.


\end{document}